\documentclass{article}
\usepackage{spconf,amsmath,graphicx}
\usepackage{amssymb}
\usepackage{commath}
\usepackage{mathtools}
\usepackage{soul, color} 
\usepackage{amsthm}
\usepackage{float}
\usepackage{multirow}
\usepackage{booktabs}
\usepackage{hyperref}
\usepackage{cleveref}

\crefformat{section}{\S#2#1#3} 
\crefformat{subsection}{\S#2#1#3}
\crefformat{subsubsection}{\S#2#1#3}

\newtheorem{theorem}{Theorem}
\newtheorem{remark}{Remark}

\let\oldnorm\norm   
\let\norm\undefined 
\DeclarePairedDelimiter\norm{\lVert}{\rVert}

\let\abs\undefined 
\DeclarePairedDelimiter\abs{\lvert}{\rvert}

\title{A DNN based Normalized Time-frequency Weighted Criterion for Robust Wideband DoA Estimation}

\name{Kuan-Lin Chen$^1$, Ching-Hua Lee$^1$, Bhaskar D. Rao$^1$, and Harinath Garudadri$^2$\thanks{This work was supported in part by NIH/NIDCD under Grant R01DC015436 and in part by NSF/IIS under Grant 1838897.}}
\address{$^1$Department of Electrical and Computer Engineering, University of California, San Diego\\$^2$Qualcomm Institute, University of California, San Diego}

\begin{document}
\ninept

\maketitle

\begin{abstract}
Deep neural networks (DNNs) have greatly benefited direction of arrival (DoA) estimation methods for speech source localization in noisy environments. However, their localization accuracy is still far from satisfactory due to the vulnerability to nonspeech interference. To improve the robustness against interference, we propose a DNN based normalized time-frequency (T-F) weighted criterion which minimizes the distance between the candidate steering vectors and the filtered snapshots in the T-F domain. Our method requires no eigendecomposition and uses a simple normalization to prevent the optimization objective from being misled by noisy filtered snapshots. We also study different designs of T-F weights guided by a DNN. We find that duplicating the Hadamard product of speech ratio masks is highly effective and better than other techniques such as direct masking and taking the mean in the proposed approach. However, the best-performing design of T-F weights is criterion-dependent in general. Experiments show that the proposed method outperforms popular DNN based DoA estimation methods including widely used subspace methods in noisy and reverberant environments.
\end{abstract}


\section{Introduction}
One of the goals in speech source localization for many applications such as hearing aids \cite{pisha2019wearable} and augmented hearing systems \cite{pisha2018wearable} is to estimate the direction of arrival (DoA) of speech signals using microphone arrays. Building upon the foundation of narrowband DoA estimation, the most widely used wideband methods are based on a coherent combination of spatial spectra \cite{mohan2003localization,mohan2008localization} or spatial covariance matrices (SCMs) \cite{wang1985coherent,hung1988focussing} at different frequencies to improve the robustness of DoA estimates against noise.
For example, the wideband MUltiple SIgnal Classification (MUSIC) \cite{schmidt1986multiple} and its variants \cite{van2004optimum,azimi2008wideband}.
However, the performance of these conventional methods degrades rapidly with a lower signal-to-noise ratio (SNR) or signal-to-interference ratio (SIR). Without reducing the noise and interference components from the \textit{snapshots} (see Section \ref{sec: problem formulation}), it is difficult to improve the estimation of SCMs of the speech for better performance.

The advent of deep learning has opened up many opportunities for conventional approaches \cite{chen2021resnests}. With DNNs, rich T-F patterns in speech and nonspeech are learned to remove noisy components from microphone measurements \cite{wang2018supervised}.
For example, in \cite{xu2017weighted,yang2017multiple,wang2018robust,yang2019multiple}, T-F weights predicted from DNNs are assigned to snapshots to better estimate speech SCMs, resulting in the so-called weighted SCMs (WSCMs). When a snapshot is more noisy, a smaller T-F weight is assigned to de-emphasize its contribution to a speech SCM estimate. Because these WSCMs contain less interference and noise components, they are more accurate speech SCMs achieving superior performance for conventional DoA estimation approaches such as wideband MUSIC and its variants.
Many other applications such as acoustic beamforming, speech enhancement, and speech recognition \cite{heymann2016neural,erdogan2016improved,xiao2017time,ochiai2017unified,wang2018all,wang2018robust_interspeech,pfeifenberger2019eigenvector} have adopted such a framework to improve the estimation of SCMs for downstream tasks.
\begin{figure}[t]
\begin{minipage}[b]{0.158\textwidth}
  \centering
  \centerline{\includegraphics[width=0.88\linewidth]{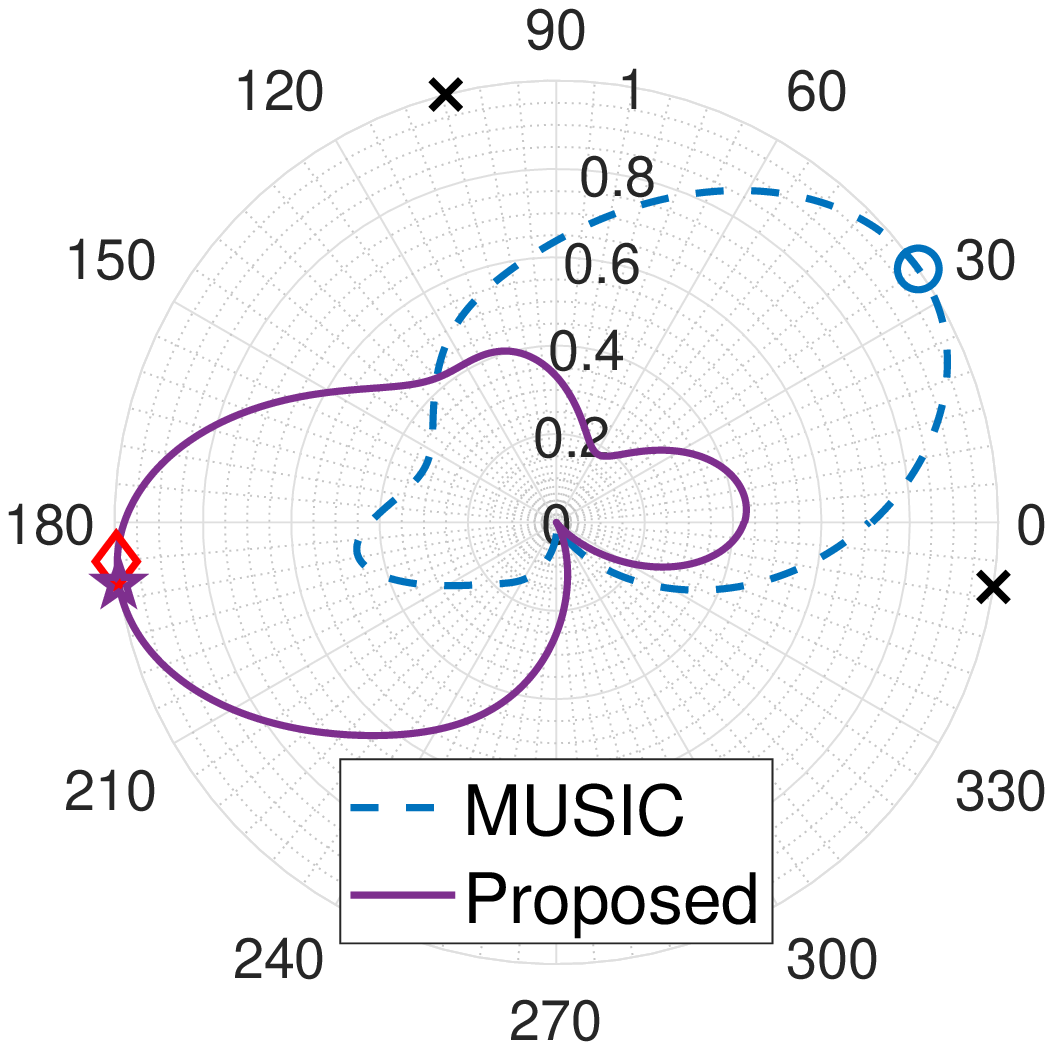}}\smallskip
  \centerline{(a) $\text{SIR}=-6$ dB.}
  \label{fig:spatial_spectrum_2}
\end{minipage}
    \begin{minipage}[b]{0.158\textwidth}
  \centering
  \centerline{\includegraphics[width=0.88\linewidth]{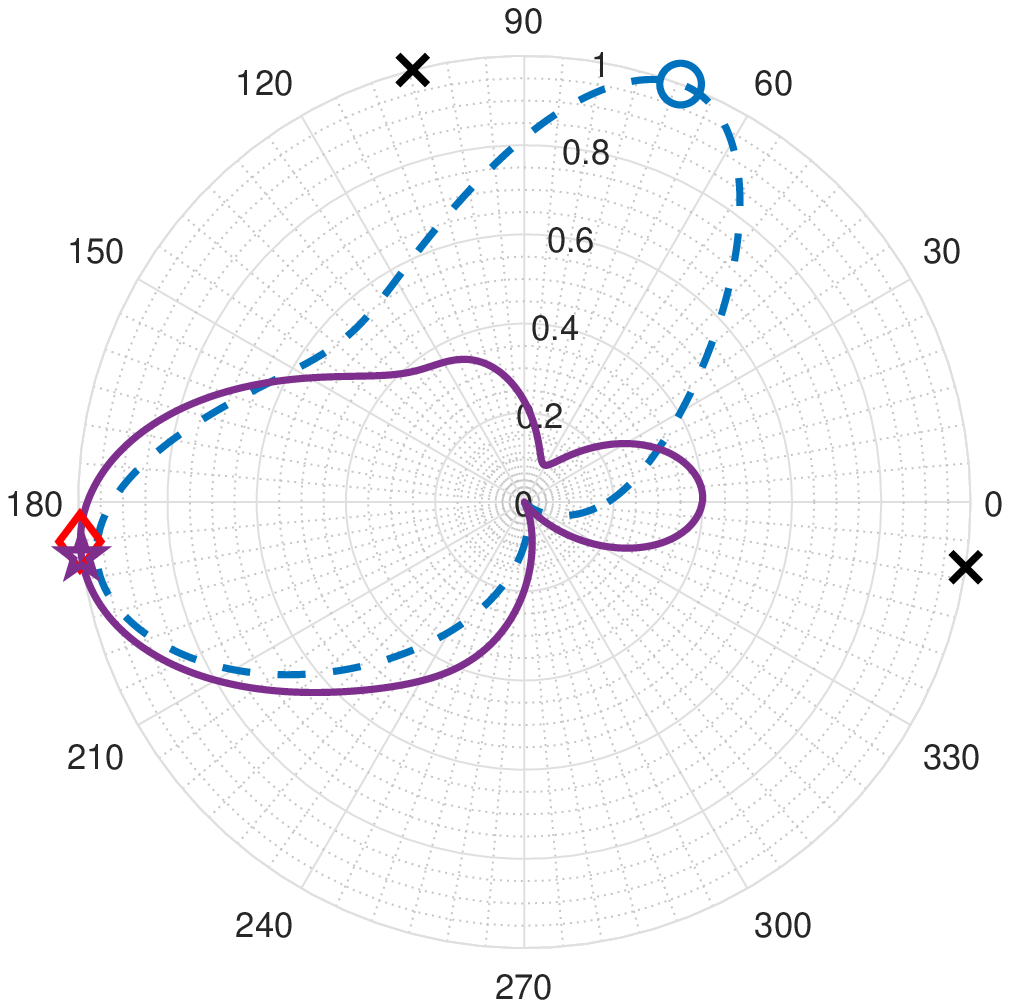}}\smallskip
  \centerline{(b) $\text{SIR}=0$ dB.}
  \label{fig:spatial_spectrum_3}
\end{minipage}
    \begin{minipage}[b]{0.158\textwidth}
  \centering
  \centerline{\includegraphics[width=0.88\linewidth]{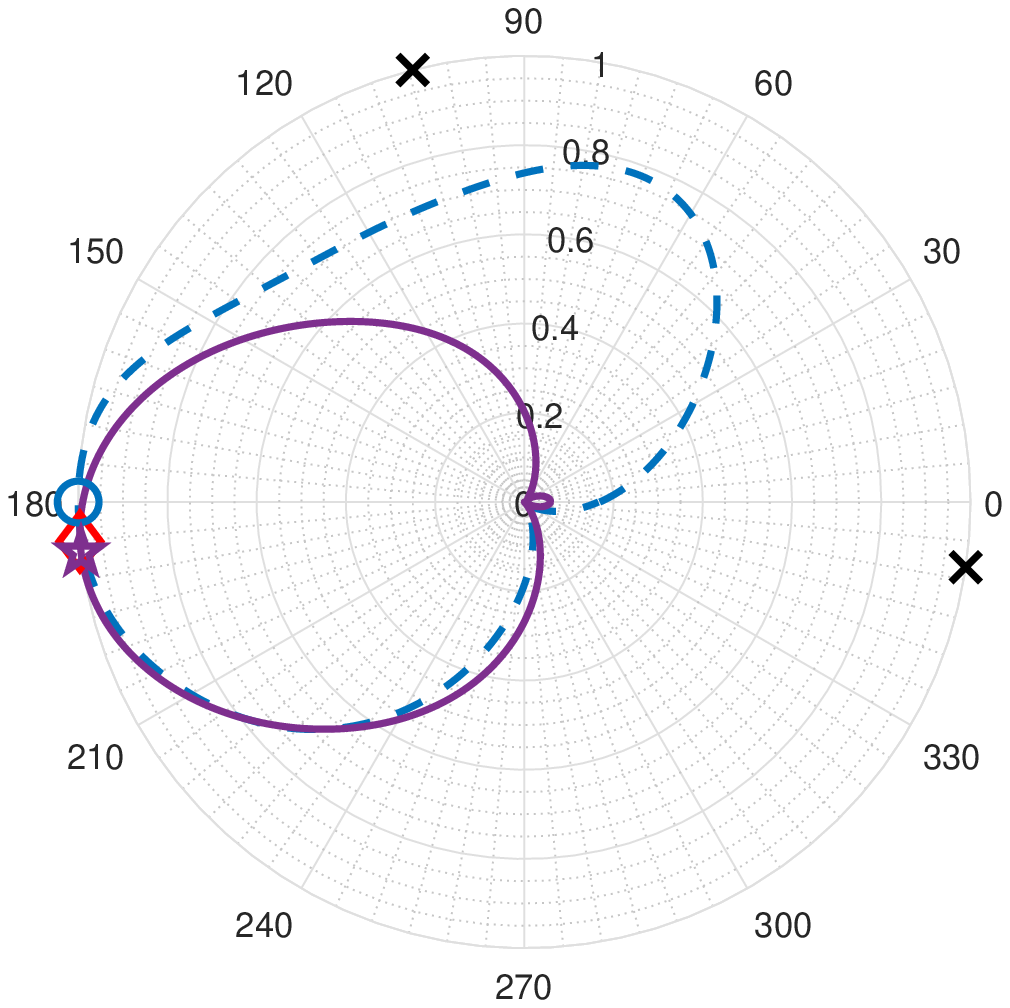}}\smallskip
  \centerline{(c) $\text{SIR}=20$ dB.}
  \label{fig:spatial_spectrum_4}
\end{minipage}
\caption{Normalized spatial pseudo-spectra. {\color{red}{$\diamondsuit$}} and $\boldsymbol{\times}$ represent the speaker and interference, respectively. The proposed method (\ref{eq:proposed_normalized_method}) is more robust than the DNN based MUSIC (\ref{eq:music_optimization}) in a wide range of SIRs.}
\label{fig:spatial_spectrum}
\end{figure}

While recent progress in the above works has demonstrated better estimation of speech SCMs using DNNs, we also find that most of them are still vulnerable to nonspeech interference, even though it is weaker than the target speech. For example, the two popular subspace approaches including the WSCM based MUSIC \cite{xu2017weighted} and principal vector method \cite{yang2017multiple,wang2018robust,yang2019multiple},
where the latter has been widely used in many neural beamforming works.
Because these methods use normalized eigenvectors of WSCMs to determine the spatial spectra at different frequencies, they lose the information encoded by eigenvalues that are known to highly correlate with SNR and SIR.
Furthermore, these methods rely on the assumption that the power of the target signal component is larger than the noise or interference, which is unlikely to be true for all frequency components even for an interference that is weaker than speech.
These robustness issues in the WSCM based subspace methods make them less attractive compared to other DNN based methods such as \cite{pertila2017robust,wang2018target}. On the other hand, a comparative study of the best-performing design of T-F weights for different algorithms still remains largely lacking.

In this paper, we propose a \textit{normalized T-F weighted criterion} and study different designs of the T-F weights guided by a DNN for robust wideband DoA estimation. Our criterion exploits all T-F weights and is more robust to noise and interference with less computational complexity as compared to existing WSCM based approaches that rely on subspaces and eigendecomposition. The proposed method can be applied to arbitrary array geometries, and the DNN used to guide T-F weights is independent of the microphone array used. The training data in our methods are easy to obtain because only single-channel speech and nonspeech corpora are required for training. Experiments show that the proposed normalized T-F weighted method outperforms the DNN based subspace approaches including the popular MUSIC and principal vector method, and a DNN based non-subspace method using the steered response power (SRP), under different types and levels of interference in noisy and reverberant environments.

\section{The Signal Model} \label{sec: problem formulation}
Let the received signal at the $M$-element microphone array in the short-time Fourier transform (STFT) domain be $\mathbf{y}(t,f)=\begin{bmatrix}y_1(t,f)&y_2(t,f)&\cdots&y_M(t,f)\end{bmatrix}^\mathsf{T}\in\mathbb{C}^{M}$ for all $(t,f)$ with the frequency bin index $f\in\{1,2,\cdots,F\}$ and time frame index $t\in\{1,2,\cdots,T\}$ where
$
        y_m(t,f)=\sum_{i=1}^Ir_{m,s_i}(f)s_i(t,f)+\sum_{k=1}^Kr_{m,h_k}(f)h_k(t,f)+n_m(t,f)
$
for $m=1,2,\cdots,M$. We have used $s_i(t,f)\in\mathbb{C}$ to denote the STFT of the time-domain target speech signal $s_i(\tau)$, $h_k(t,f)\in\mathbb{C}$ to denote the STFT of the time-domain nonspeech interference $h_k(\tau)$, and $n_m(t,f)\in\mathbb{C}$ to denote the STFT of the time-domain additive noise on the $m$-th microphone. $I$ and $K$ are nonnegative integers representing the number of speech sources and interference sources, respectively. $r_{m,s_i}(f)\in\mathbb{C}$ denotes the acoustic transfer function (ATF) between the location of the $i$-th target and microphone $m$. $r_{m,h_k}(f)\in\mathbb{C}$ denotes the ATF between the location of the $k$-th interference and microphone $m$. The vector $\mathbf{y}(t,f)$ is referred to as the \textit{snapshot} in the literature. $\odot$ denotes the Hadamard product. Let $[M]=\{1,2,\cdots,M\}$.
\section{Optimization Criteria in Prior Work}
In this section, we review two popular wideband DoA estimation approaches and point out their issues.
Define
$\mathbf{w}(t,f)=\begin{bmatrix}w_1(t,f)&w_2(t,f)&\cdots&w_M(t,f)\end{bmatrix}^{\mathsf{T}}$ where $w_m(t,f)$ is an estimate of the ideal ratio mask (IRM) for the $m$-th microphone. Under the noise subspace assumption, the popular DNN based MUSIC \cite{xu2017weighted} finds the DoA by solving the following optimization problem
\begin{equation} \label{eq:music_optimization}
    \max_{\theta}\ \ \ \sum_{f}\frac{1}{\mathbf{v}^{\mathsf{H}}(\theta,f)\mathbf{N}(f)\mathbf{N}^{\mathsf{H}}(f)\mathbf{v}(\theta,f)}
\end{equation}
where $\mathbf{N}(f)$ represents the noise subspace formed by the eigenvectors corresponding to the $M-1$ smallest eigenvalues of the WSCM
\begin{equation}
    \boldsymbol{\Phi}(f)=\sum_{t}\left[\mathbf{w}(t,f)\odot\mathbf{y}(t,f)\right]\left[\mathbf{w}(t,f)\odot\mathbf{y}(t,f)\right]^{\mathsf{H}}.
\end{equation}
The DNN-guided T-F weights
$
    \mathbf{w}(t,f)
$
are introduced to filter the snapshot at every T-F bin. The weights $w_m(t,f)$ are predicted by a DNN model followed by a post-processing technique before solving the optimization problem for $\theta$. We describe the prediction and learning of these T-F weights in Section \ref{subsection:t-f-weights}. We have used $\mathbf{v}(\theta,f)\in\mathbb{C}^{M}$ to denote the array manifold (or steering vector) at DoA $\theta$ and frequency corresponding to the bin index $f$ under the plane wave assumption \cite{van2004optimum}.
Another popular approach \cite{yang2017multiple,wang2018robust,yang2019multiple} is the principal vector method relying on the signal subspace, which finds the DoA by solving the following optimization problem
\begin{equation} \label{eq:principal_vector_optimization}
    \max_{\theta}\ \ \ \sum_{f}\mathbf{v}^{\mathsf{H}}(\theta,f)\mathbf{p}(f)\mathbf{p}^{\mathsf{H}}(f)\mathbf{v}(\theta,f)
\end{equation}
where $\mathbf{p}(f)$ is the principal eigenvector of the WSCM $\boldsymbol{\Phi}(f)$. (\ref{eq:music_optimization}) and (\ref{eq:principal_vector_optimization}) are common subspace approaches. Remark \ref{remark:power_assumption} and \ref{remark:losing_weights} point out their weaknesses. In terms of computational complexity, computing the signal or noise subspace requires performing eigenvalue decomposition which is computationally expensive when $M$ is large.
\begin{remark} \label{remark:power_assumption}
    The power of the speech signal is assumed to be larger than the power of undesired signals. When the interference and noise are stronger than the speech signal at some frequencies, such an assumption can create misleading spatial spectra at those frequencies.
\end{remark}
\begin{remark} \label{remark:losing_weights}
    Because $\mathbf{w}(t,f)$ suppresses the interference and noise, the principal eigenvalue of an SCM at a noisy frequency is reduced. Taking normalized eigenvectors of every WSCM from different frequencies and combining them uniformly in (\ref{eq:music_optimization}) or (\ref{eq:principal_vector_optimization}) weight every frequency equally, losing the ability to de-emphasize the interference and noise components according to the eigenvalues in the objective.
\end{remark}
\section{Proposed Methods} \label{section:weighted_criteria}
To overcome the limitations revealed by Remarks \ref{remark:power_assumption} and \ref{remark:losing_weights}, one can use a non-subspace method. The most intuitive approach is a DNN based SRP method solving the following optimization problem
\begin{equation} \label{eq:srp_method_optimization_max_cov}
    \max_{\theta}\ \ \ \sum_{f}\mathbf{v}^{\mathsf{H}}(\theta,f)\boldsymbol{\Phi}(f)\mathbf{v}(\theta,f)
\end{equation}
in which it uses the whole covariance matrix $\boldsymbol{\Phi}(f)$. Such an approach is free from the power assumption and picking eigenvectors as a signal or noise subspace. However, it heavily relies on the estimation quality of T-F weights $\mathbf{w}(t,f)$. Since (\ref{eq:srp_method_optimization_max_cov}) is not robust to imperfect $\mathbf{w}(t,f)$, it is still vulnerable to interference and noise. 
\subsection{A Robust DNN based Normalized T-F Weighted Criterion} \label{sec:normalized_weighted_criterion}
When the magnitude of the measurement $\mathbf{y}(t,f)$ is large with the corresponding $w_m(t,f)$ being not sufficiently small for some $m$, a normalization of the magnitude of $\mathbf{y}(t,f)$ can prevent the objective function from relying on a single low SNR or SIR snapshot (outlier) to a certain degree. Based on such a rationale, we first normalize the filtered snapshot at every T-F bin and then directly match a candidate steering vector to the normalized filtered snapshot, giving the following DNN based normalized T-F weighted criterion
\begin{equation} \label{eq:proposed_method_1_optimization}
    \min_{\theta,\mathbf{S}}\ \ \ \sum_{f}\sum_{t}\oldnorm{\frac{\mathbf{w}(t,f)\odot\mathbf{y}(t,f)}{\norm{\mathbf{y}(t,f)}_{2}}-s(t,f)\mathbf{v}(\theta,f)}_2^2
\end{equation}
where the optimization variable matrix $\mathbf{S}\in\mathbb{C}^{T\times F}$ is the STFT of the unknown target speech signal whose $(t,f)$ element is $s(t,f)$.
Note that it is reasonable to assume $\norm{\mathbf{y}(t,f)}_{2}>0$ but we have made a strong assumption that $I=1$.
Although $\theta$ and $\mathbf{S}$ are two different unknowns, we are only interested in $\theta$, which can be found by peak detection without knowing $\mathbf{S}$ by Theorem \ref{main_theorem}.
\begin{theorem} \label{main_theorem}
Let $\left(\theta^*,\mathbf{S}^*\right)$ be any local minimizer of (\ref{eq:proposed_method_1_optimization}) and $\tilde{\mathbf{y}}(t,f)=\mathbf{w}(t,f)\odot\mathbf{y}(t,f)$. Then, $\theta^*$ is a local maximizer of
\begin{equation} \label{eq:proposed_normalized_method}
    \begin{split}
    \max_{\theta}\ \ \ \sum_{f}\mathbf{v}^{\mathsf{H}}(\theta,f)\sum_{t}\frac{\tilde{\mathbf{y}}(t,f)\tilde{\mathbf{y}}^{\mathsf{H}}(t,f)}{\norm{\mathbf{y}(t,f)}_2^2}\mathbf{v}(\theta,f).
    \end{split}
\end{equation}
Also, any local maximizer of (\ref{eq:proposed_normalized_method}) is a $\theta^*$.
\end{theorem}
\begin{proof}
$\theta^*$ is equivalent to a local maximizer of the objective
$
        \sum_{f}\sum_{t}\norm{\mathbf{v}(\theta,f)\mathbf{v}^{\dagger}(\theta,f)\frac{\mathbf{w}(t,f)\odot\mathbf{y}(t,f)}{\norm{\mathbf{y}(t,f)}_2}}_2^2
$
where $\mathbf{v}^{\dagger}(\theta,f)$ is the pseudoinverse of $\mathbf{v}(\theta,f)$. The spatial pseudo-spectrum at $f$ is given by
$
        \text{tr}\left(\mathbf{A}^{\mathsf{H}}(\theta,f)\mathbf{A}(\theta,f)\sum_{t}\frac{\tilde{\mathbf{y}}(t,f)\tilde{\mathbf{y}}^{\mathsf{H}}(t,f)}{\norm{\mathbf{y}(t,f)}_2^2}\right)
$ where $\mathbf{A}(\theta,f)=\mathbf{v}(\theta,f)\mathbf{v}^{\dagger}(\theta,f)$, giving (\ref{eq:proposed_normalized_method}) with the cyclic property of the trace.
\end{proof}
\subsection{Design of T-F Weights: DNN and Post-processing} \label{subsection:t-f-weights}
To make the DNN predicting the T-F weights $\mathbf{w}(t,f)$ independent of the array geometry used in the wideband DoA estimation problem, we propose a unified framework consisting of a signal enhancement DNN model and a post-processing technique.
Let $\mathbf{W}_m$, $\mathbf{G}_m$, and $\mathbf{Y}_m$ be $T$-by-$F$ matrices whose $(t,f)$-th element are $w_m(t,f)$, $G_m(t,f)$, and $y_m(t,f)$, respectively. For all $m\in\{1,2,\cdots,M\}$, the DNN $g:\mathbb{R}^{2\times T\times F}\to\mathbb{R}^{T\times F}$ individually predicts a weight matrix $\mathbf{G}_m$ by using the raw T-F representation $\mathbf{Y}_m\in\mathbb{C}^{T\times F}$ obtained at the $m$-th microphone. Once all the weight matrices are computed, we feed them to a post-processing function $q_m:\mathbb{R}^{M\times T\times F}\to\mathbb{R}^{T\times F}$ to generate the final T-F weights $w_m(t,f)$ for the $m$-th microphone. Mathematically,
$
    \mathbf{W}_m=q_m\left(\mathbf{G}_1,\mathbf{G}_2,\cdots,\mathbf{G}_M\right)
$
where $\mathbf{G}_m$ is computed independently by $g$, i.e.,
$
    \mathbf{G}_m = g\left(\Re\{\mathbf{Y}_m\},\Im\{\mathbf{Y}_m\}\right)
$.  We have used $\Re$ and $\Im$ to extract the real and imaginary parts of a matrix, respectively. We can utilize a different post-processing for $q_m$ as shown in Table \ref{tab:post_processing_techniques}. The above composition of the post-processing and the DNN establishes the design of T-F weights. During training, only single-channel speech and nonspeech data are required and the DNN model is trained to learn the ideal ratio mask (IRM) \cite{wang2018supervised,wang2018robust}.
\begin{table}[h]
    \centering
    \caption{Examples of the post-processing function $q_m$.}
    \begin{tabular}{ll}
        \toprule
        Post-processing & Expression for all $m\in[M]$ \\
        \midrule
        Identity (direct masking) &  $q_m=\mathbf{G}_m$\\
        Minimum & $[q_m]_{t,f}=\min_{i\in[M]}[\mathbf{G}_i]_{t,f}$\\
        Maximum & $[q_m]_{t,f}=\max_{i\in[M]}[\mathbf{G}_{i}]_{t,f}$\\
        Arithmetic mean & $q_m=\frac{1}{M}\sum_{i=1}^M\mathbf{G}_i$\\
        Arithmetic median & $[q_m]_{t,f}=\text{median}(\{[\mathbf{G}_i]_{t,f}\}_{i=1}^M)$\\
        Hadamard product & $q_m=\mathbf{G}_1\odot\mathbf{G}_2\odot\cdots\odot\mathbf{G}_M$\\
        Geometric mean & $[q_m]_{t,f}=\sqrt[M]{\prod_{i=1}^M[\mathbf{G}_i]_{t,f})}$\\
        \multirow{2}{*}{Binary thresholding (BT)} & $[q_m]_{t,f}=1,$ if $[\mathbf{G}_m]_{t,f}>\beta$\\
         & $[q_m]_{t,f}=0,$ otherwise\\
        \bottomrule
    \end{tabular}
    \label{tab:post_processing_techniques}
\end{table}
\section{Experiments} \label{section:simulation_results}
The experiments are conducted at $16$ kHz with the target speech from the TIMIT dataset \cite{garofolo1993timit}, interference from the PNL $100$ nonspeech sounds \cite{hu2010tandem} (machine, water, wind, etc), and white Gaussian noise. We randomly divide the $100$ nonspeech sounds into $80$ for training and $20$ for testing. Code is available at \url{https://github.com/kjason/DnnNormTimeFreq4DoA}.

\noindent\textbf{The DNN:}
We use a U-Net \cite{ronneberger2015u} ($0.67$M parameters) to predict the IRM. The logistic sigmoid function is applied to the last layer of the U-Net.
The loss function is an absolute error loss. The batch size is $16$. All networks are trained for $200$ epochs. SGD with Nesterov momentum is used. The momentum is set to $0.9$. The weight decay is $0.0005$. The learning rate is initially set to $0.1$ and decreased by a factor of $5$ after training $60$, $120$, and $160$ epochs. The speech and interference are created by mixing $P$ clean speech files and $P$ interference files, respectively. $P\in\{1,2,3\}$ is uniformly sampled in each example and each file is uniformly sampled from the training set.
The SNR and SIR are uniformly sampled in the range of $[0,20]$ and $[-10,20]$, respectively. We use PyTorch \cite{paszke2019pytorch} to train the DNN.

\begin{figure}[t]
\begin{minipage}[b]{0.24\textwidth}
  \centering
  \centerline{\includegraphics[width=0.99\linewidth]{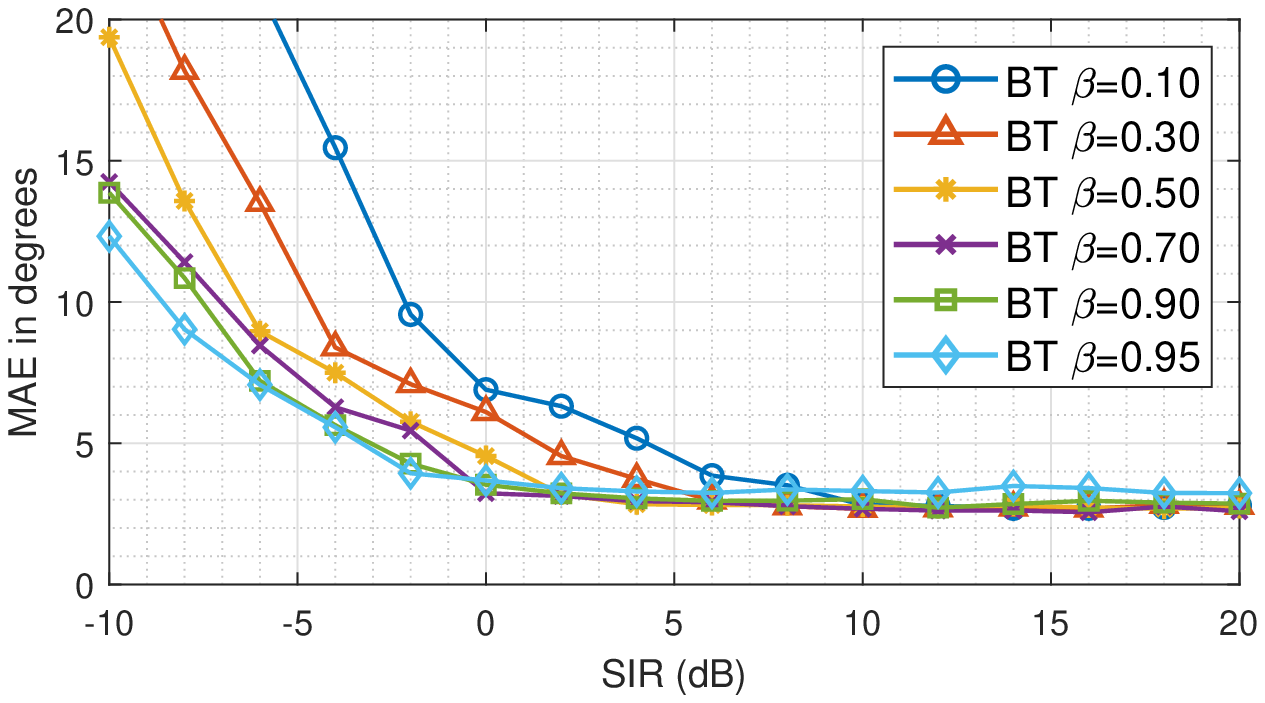}}\smallskip
  \centerline{(a) BT with different $\beta$.}
  \label{fig:post_processing_music_thresholding}
\end{minipage}
\begin{minipage}[b]{0.24\textwidth}
  \centering
  \centerline{\includegraphics[width=0.99\linewidth]{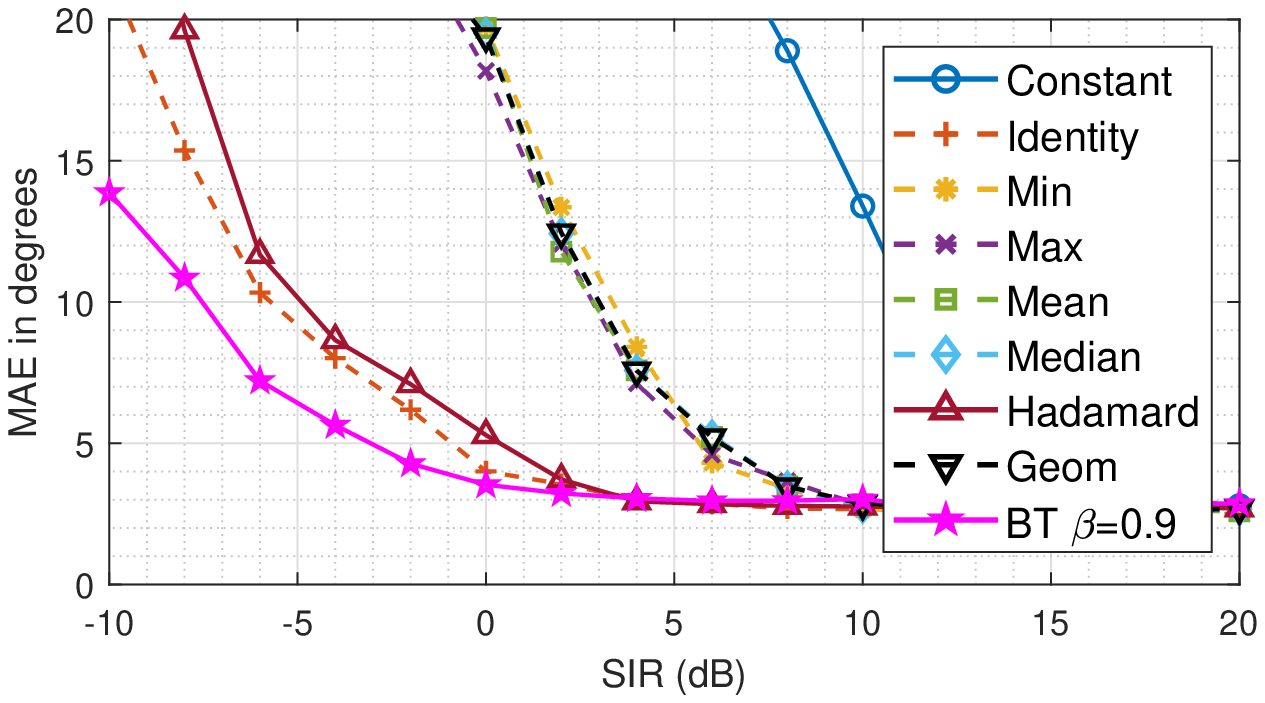}}\smallskip
  \centerline{(b) Overall comparison.}
  \label{fig:post_processing_music}
\end{minipage}
    \caption{MAE in degrees vs. SIR (the same setting as Fig. \ref{fig:post_processing}). Different post-processing functions are evaluated for the DNN based MUSIC (\ref{eq:music_optimization}). ``Constant'' means $w_m(t,f)=1,\forall (m,t,f)$, leading to original sample SCMs (the signal enhancement model is not used).}
    \label{fig:post_processing_music_thresholding_and_all}
\end{figure}
\begin{figure*}[!ht]
\begin{minipage}[b]{0.33\textwidth}
  \centering
  \centerline{\includegraphics[width=0.99\linewidth]{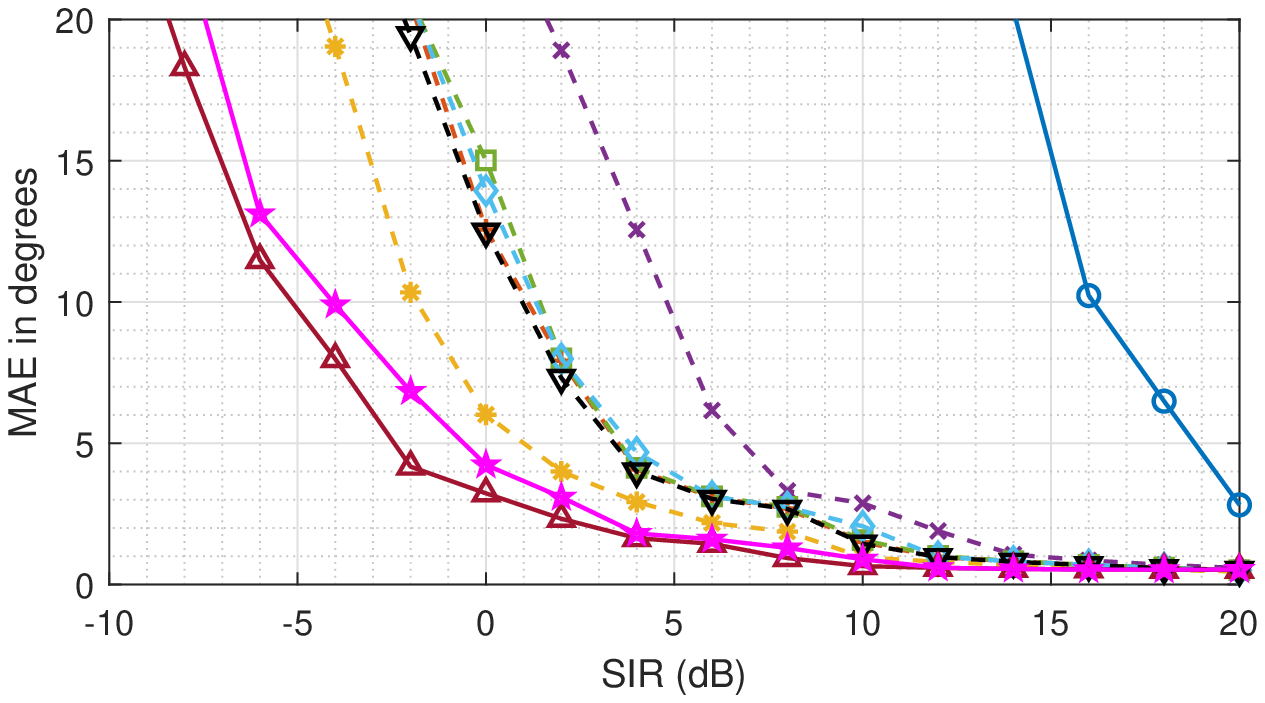}}\smallskip
  \centerline{(a) The proposed method (\ref{eq:proposed_normalized_method}).}
  \label{fig:post_processing_normalized_weighted}
\end{minipage}
\begin{minipage}[b]{0.33\textwidth}
  \centering
  \centerline{\includegraphics[width=0.99\linewidth]{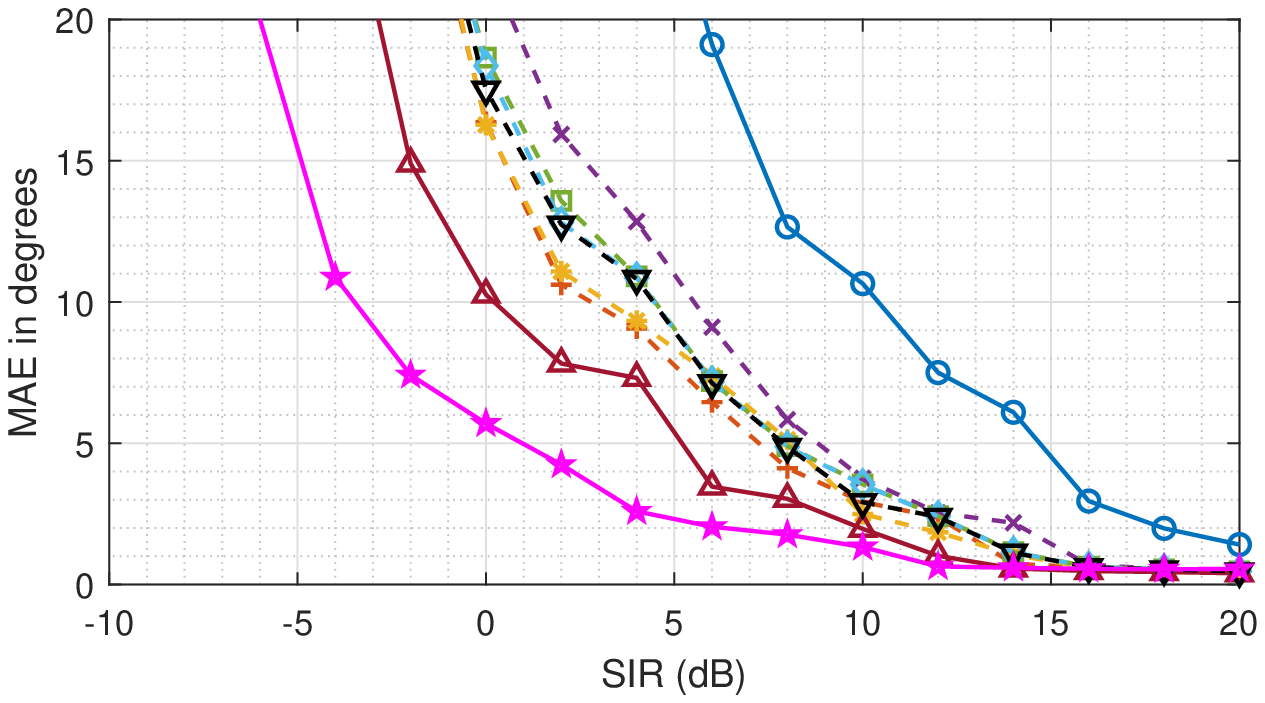}}\smallskip
  \centerline{(b) The principal vector method (\ref{eq:principal_vector_optimization}).}
  \label{fig:post_processing_principal}
\end{minipage}
\begin{minipage}[b]{0.33\textwidth}
  \centering
  \centerline{\includegraphics[width=0.99\linewidth]{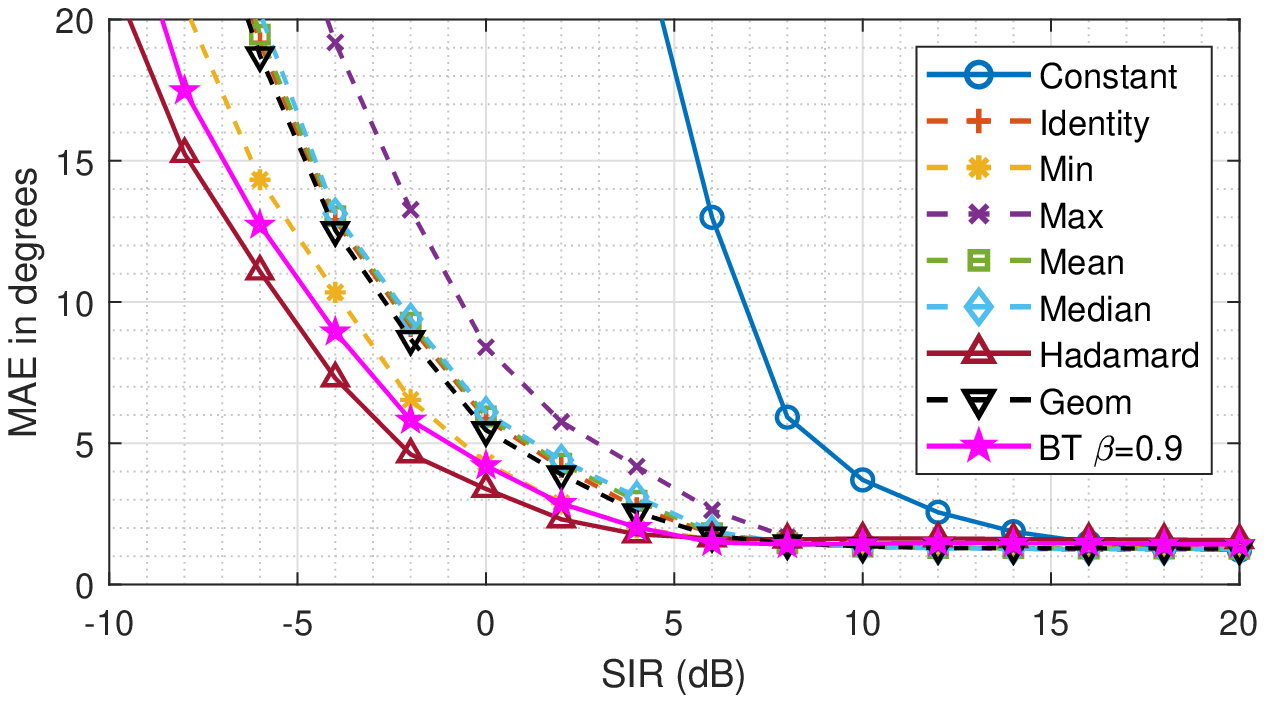}}\smallskip
  \centerline{(c) The SRP method (\ref{eq:srp_method_optimization_max_cov}).}
  \label{fig:post_processing_srp}
\end{minipage}
\caption{MAE in degrees vs. SIR ($K=1$, $\text{RT}_{60}=0.3$s, and $\text{SNR}=20$ dB). Different post-processing functions are evaluated for different DNN based methods. Note that the ranking of these post-processing methods depends on the DoA estimation algorithm (also see Fig. \ref{fig:post_processing_music_thresholding_and_all}).}
\label{fig:post_processing}
\end{figure*}
\begin{figure*}[!htb]
\begin{minipage}[b]{0.33\textwidth}
  \centering
  \centerline{\includegraphics[width=0.99\linewidth]{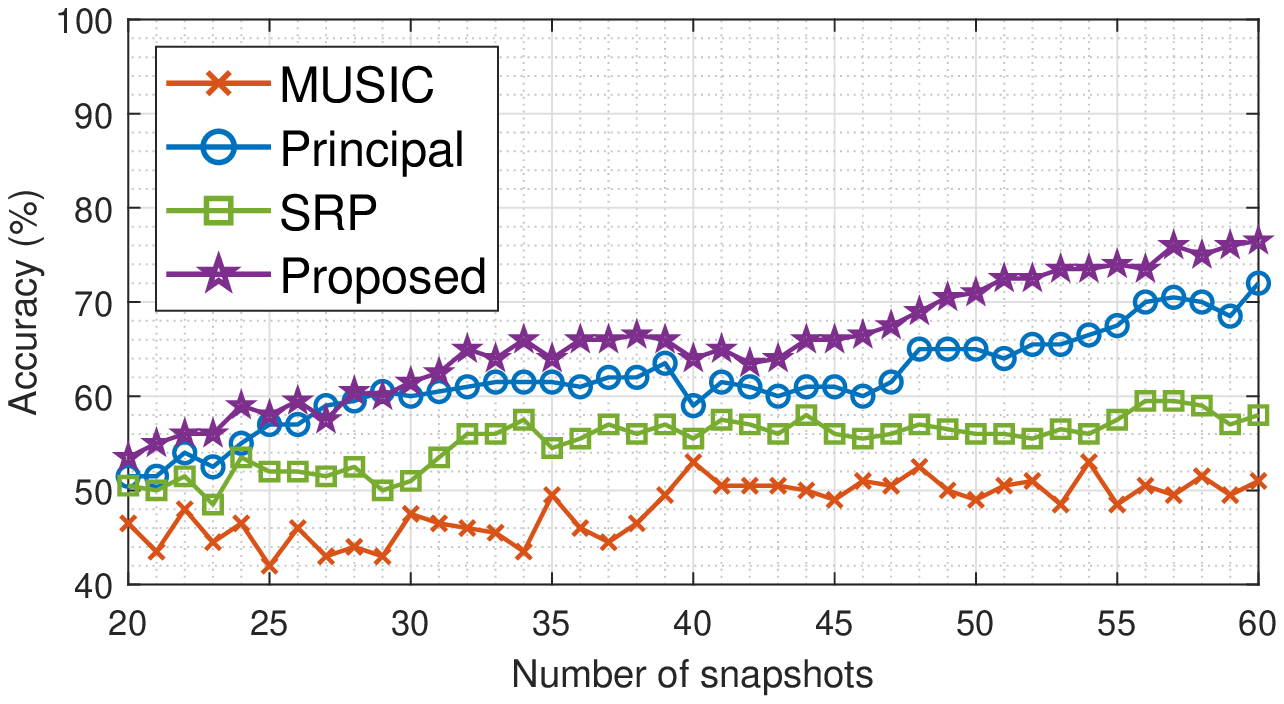}}\smallskip
  \centerline{(a) $-6$ dB SIR.}
  \label{fig:n_snapshots_sir_minus6}
\end{minipage}
\begin{minipage}[b]{0.33\textwidth}
  \centering
  \centerline{\includegraphics[width=0.99\linewidth]{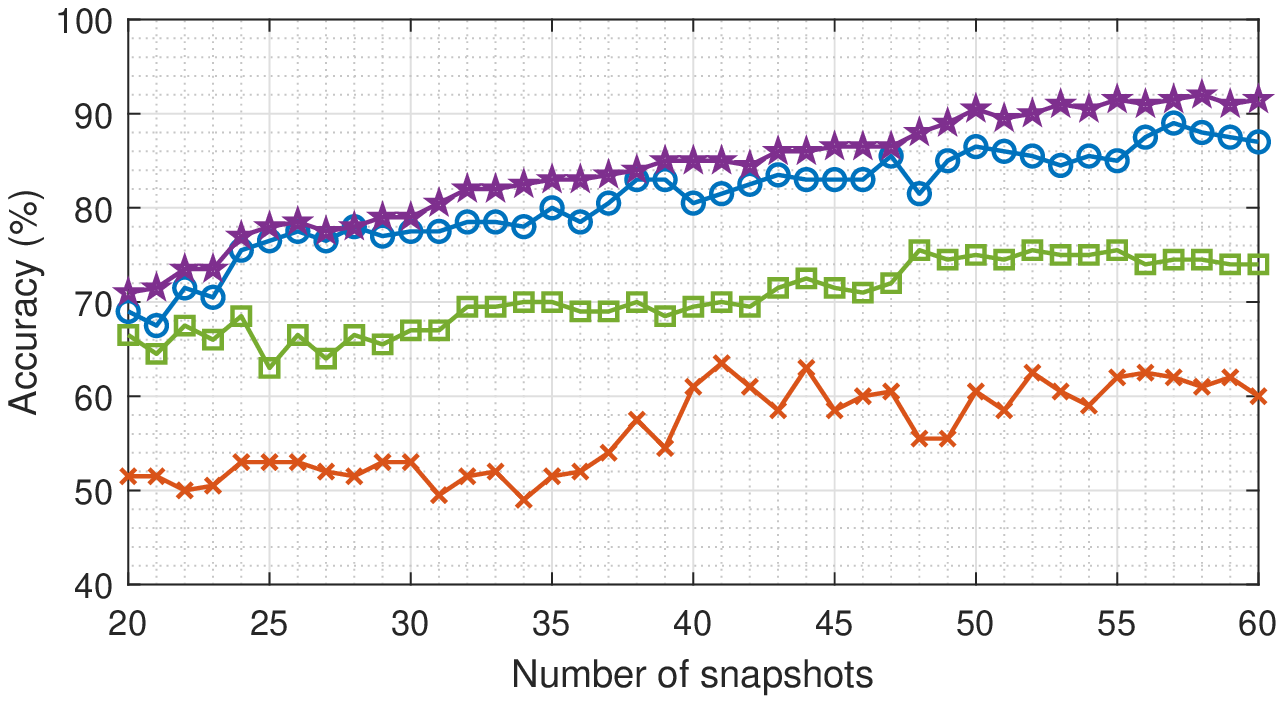}}\smallskip
  \centerline{(b) $0$ dB SIR.}
  \label{fig:n_snapshots_sir_0}
\end{minipage}
\begin{minipage}[b]{0.33\textwidth}
  \centering
  \centerline{\includegraphics[width=0.99\linewidth]{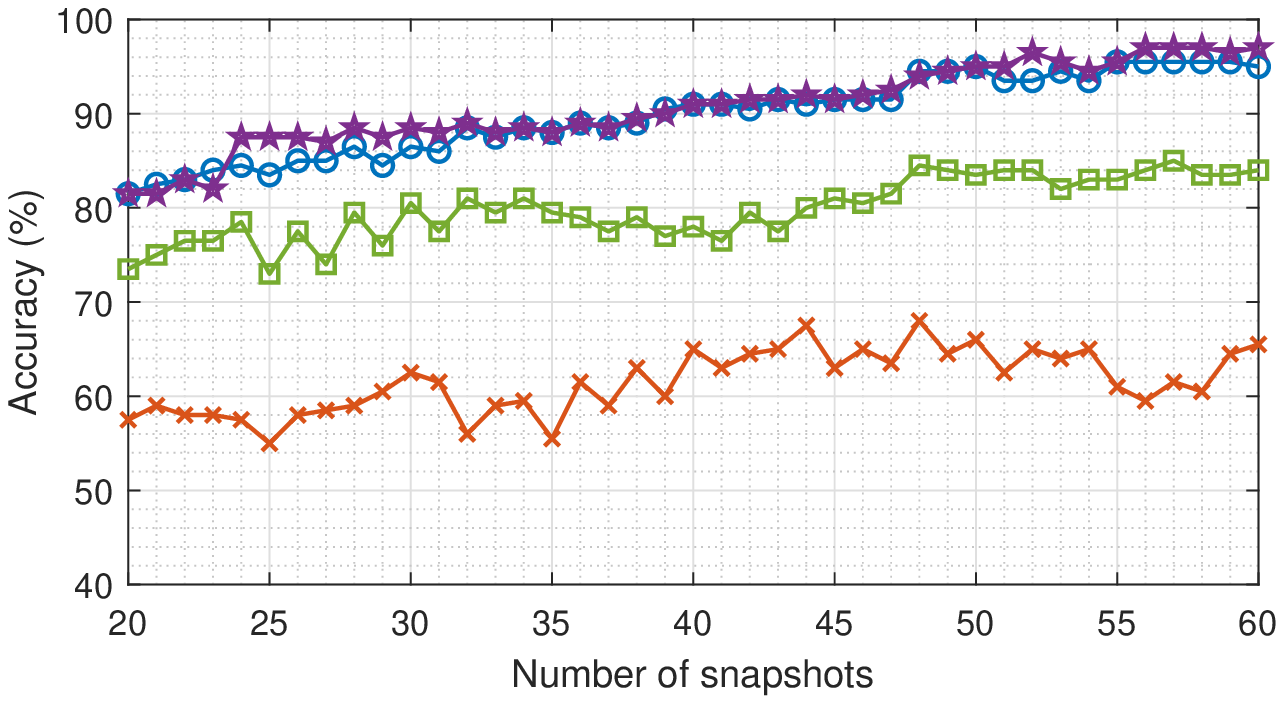}}\smallskip
  \centerline{(c) $+6$ dB SIR.}
  \label{fig:n_snapshots_sir_plus6}
\end{minipage}
\caption{Accuracy vs. number of snapshots $T$ ($K=1$, $\text{RT}_{60}=0.3$s, and $\text{SNR}=20$ dB). The proposed method outperforms all baselines.}
\label{fig:n_snapshots}
\end{figure*}
\noindent\textbf{Room acoustics and other settings:}
A $3$D room with $9.0$m (x-axis), $7.0$m (y-axis), and $3.5$m (z-axis) is used.
A $9$-element rectangular microphone array paralleled to the xy-plane with an equal spacing $0.02$m along the x-axis and y-axis is used.
The center of the microphone array is placed at the center of the room $(4.50\text{m}, 3.50\text{m}, 1.75\text{m})$. The speed of sound is $343$ m/s. For every simulation using specified $\text{RT}_{60}$, SNR, SIR, $T$, $I$, and $K$, we run $C=200$ trials to compute the mean absolute error (MAE) or accuracy. The accuracy is given by the number of successes divided by $C$. A trial is rated as a success if $\abs{\hat{\theta}-\theta_{\text{gt}}}<\theta_{\text{threshold}}=3^\circ$.
In each trial, we uniformly place target speakers and nonspeech interference sources at random in the room simulating a dining environment.
The distance $r$ between a source (target or interference) and the center of the microphone array on the xy-plane is uniformly sampled in the range of $1.0$m to $3.0$m. The height $z$ of the source is uniformly sampled between $1.0$m and $1.8$m. The azimuthal angle $\theta_{\text{gt}}$ (the ground truth DoA) is uniformly sampled between $0^\circ$ and $360^\circ$.
The angle between any two sources is at least $\theta_{\text{min}}=10^\circ$ apart. For every trial, we uniformly sample $I$ clean speech files and $K$ interference files from the test sets. Pyroomacoustics \cite{scheibler2018pyroomacoustics} is used to generate the RIRs based on the image source model \cite{allen1979image} and detect peaks. The resolution in the $1$D grid search is $0.5^\circ$. We assume only $1$ speaker in the room, i.e., $I=1$, but allow multiple interference sources, i.e., $K\geq 1$. $T=50$ if not explicitly specified. $1024$-point FFT is used. Frequency bins corresponding to $50$ Hz to $7$ kHz are used because this is the frequency band of wideband speech coders \cite{cox2009itu}.

\noindent\textbf{Baseline methods:}
The baselines are the DNN based MUSIC (\ref{eq:music_optimization}), principal vector method (\ref{eq:principal_vector_optimization}), and SRP (\ref{eq:srp_method_optimization_max_cov}). We use the same DNN and their best-performing post-processing techniques for comparison.

\noindent\textbf{Post-processing methods:}
We consider every example in Table \ref{tab:post_processing_techniques}.

\subsection{Different Post-processing Techniques for T-F Weights}
Fig. \ref{fig:post_processing_music_thresholding_and_all}(a) shows that a larger $\beta$ in the BT post-progressing gives a smaller MAE in MUSIC, indicating that the performance of DoA estimation can be improved by preserving only very high-quality snapshots. However, a higher $\beta$ reduces the number of rank-one matrices in the WSCM, potentially resulting in a singular WSCM. $\beta=0.9$ and $\beta=0.95$ give similar performance and we chose $\beta=0.9$ as the best parameter for BT. Fig. \ref{fig:post_processing_music_thresholding_and_all}(b) shows the overall comparison of different post-processing techniques. The BT gives noticeable improvements over all the other post-processing techniques for MUSIC. For the proposed normalized T-F weighted method, Fig. \ref{fig:post_processing}(a) shows that the Hadamard product is slightly better than the BT and gives the best performance. Fig. \ref{fig:post_processing}(b) and \ref{fig:post_processing}(c) show that the BT and the Hadamard product are the best-performing post-processing for the principal vector method and the SRP method, respectively. Notice that these results imply that the best post-processing of T-F weights depends on the DoA algorithm. We use the best-performing post-processing for each method in the following investigation.
\subsection{Robustness against a Wide Range of SIRs}
Fig. \ref{fig:spatial_spectrum} compares MUSIC and the proposed normalized T-F weighted method under $\text{RT}_{60}=0.9$s and $\text{SNR}=20$ dB. It shows that MUSIC can be easily misled by the interference and our proposed method is able to accurately localize the speaker in a wide range of SIRs. Table \ref{tab:compare_proposed_methods} shows the accuracy for all the candidate methods under different $\text{RT}_{60}$ and SIRs in an environment that has two interference sources. The proposed method strongly outperforms the other methods in all cases. When the environment is more reverberant, the performance degradation of the proposed method is less than the others.
\begin{table}[t]
    \caption{DoA estimation accuracy. $K=2$. $\text{SNR}=20$ dB.}
    \label{tab:compare_proposed_methods}
    \centering
\setlength{\tabcolsep}{5.5pt} 
\begin{tabular}{rcccccc}
\toprule
$\text{RT}_{60}$ (seconds)& \multicolumn{3}{c}{0.3}&\multicolumn{3}{c}{0.9}\\
SIR (dB)&$-6$&$0$&$+6$&$-6$&$0$&$+6$\\
\midrule
MUSIC (eq. \ref{eq:music_optimization})&$40\%$&$52\%$&$59\%$&$30\%$&$30\%$&$33\%$\\
Principal (eq. \ref{eq:principal_vector_optimization}) &$43\%$&$77\%$&$89\%$&$51\%$&$70\%$&$79\%$\\
SRP (eq. \ref{eq:srp_method_optimization_max_cov}) &$33\%$&$59\%$&$75\%$&$28\%$&$37\%$&$40\%$\\
Proposed  (eq. \ref{eq:proposed_normalized_method}) &$\textbf{54\%}$&$\textbf{81\%}$&$\textbf{91\%}$&$\textbf{59\%}$&$\textbf{76\%}$&$\textbf{88\%}$\\
\bottomrule
\end{tabular}
\end{table}
\subsection{Number of Snapshots}
In contrast to the proposed method, Fig. \ref{fig:n_snapshots} shows that the accuracy of MUSIC is barely improved with more snapshots, which implies that the noise subspace is hardly improved with more snapshots. On the other hand, the principal vector method performs much better, implying that the signal subspace is a better choice here. The SRP method is better than MUSIC but underperforms the principal vector method. Lastly, Fig. \ref{fig:n_snapshots} shows that the proposed method outperforms all baselines. These results hold for a wide range of numbers of snapshots $T$ and different SIRs, demonstrating consistency.
\subsection{A Closer Look at the Proposed Method}
Fig. \ref{fig:proposed_method} focuses on the evaluation of
the proposed normalized T-F weighted method in different conditions
including $\text{RT}_{60}$ and SNRs. In general, the accuracy improves
with a higher SNR or lower $\text{RT}_{60}$. The accuracy can still
approximately achieve more than 80\% for $\text{RT}_{60}=0.9s$ when both
the SIR and SNR are larger than 10 dB.
\begin{figure}[!ht]
  \centering
  \includegraphics[width=0.80\linewidth]{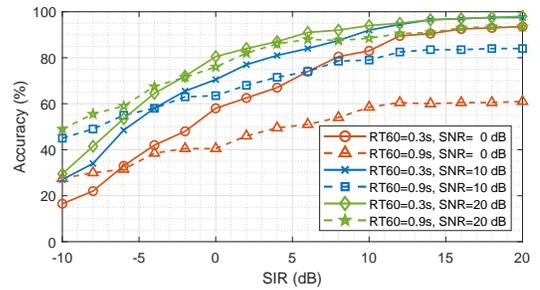}
\caption{Evaluation of the proposed method. $K=2$.}
\label{fig:proposed_method}
\end{figure}
\section{Conclusion}
To improve the robustness of WSCM based DoA estimation, we propose a normalized T-F weighted criterion that normalizes and filters snapshots to prevent the optimization objective from being misled by nonspeech components. Experimental results show that the proposed criterion outperforms MUSIC, the principal vector method, and SRP method in different noisy and reverberant environments that contain nonspeech interference sources. Given that T-F weights are crucial to performance, we also study different post-processing techniques and find that the best design is criterion-dependent.

\bibliographystyle{IEEEbib_abbr}
\bibliography{refs}

\end{document}